%% file: main.tex
\documentclass[sigconf, nonacm]{aamas}

\usepackage{balance} % for balancing columns on the final page
\usepackage{cleveref}  % for \Cref{label}
\usepackage{thm-restate} % for the restatable environment
\usepackage{multicol}

\setcopyright{ifaamas}
\acmConference[AAMAS '25]{Proc.\@ of the 24th International Conference
on Autonomous Agents and Multiagent Systems (AAMAS 2025)}{May 19 -- 23, 2025}
{Detroit, Michigan, USA}{A.~El~Fallah~Seghrouchni, Y.~Vorobeychik, S.~Das, A.~Nowe (eds.)}
\copyrightyear{2025}
\acmYear{2025}
\acmDOI{}
\acmPrice{}
\acmISBN{}

\acmSubmissionID{16}

\title[Adapting Beyond the Depth Limit]{Adapting Beyond the Depth Limit:\texorpdfstring{\\}{ }Counter Strategies in Large Imperfect Information Games}

\author{David Milec}
\affiliation{
  \institution{AI Center, FEE, CTU in Prague}
  \country{Czech Republic}}
\email{milecdav@fel.cvut.cz}

\author{Vojtěch Kovařík}
\affiliation{
  \institution{AI Center, FEE, CTU in Prague}
  \country{Czech Republic}}
\email{vojta.kovarik@gmail.com}

\author{Viliam Lisý}
\affiliation{
  \institution{AI Center, FEE, CTU in Prague}
  \country{Czech Republic}}
\email{viliam.lisy@agents.fel.cvut.cz}

\begin{abstract}
We study the problem of adapting to a known sub-rational opponent during online play while remaining robust to rational opponents.
We focus on large imperfect-information (zero-sum) games, which makes it impossible to inspect the whole game tree at once and necessitates the use of depth-limited search.
However, all existing methods assume rational play beyond the depth-limit,
    which only allows them to adapt a very limited portion of the opponent's behaviour.
We propose an algorithm \textit{Adapting Beyond Depth-limit} (\ABDshort{})
    that uses a strategy-portfolio approach -- which we refer to as matrix-valued states -- for depth-limited search.
This allows the algorithm to fully utilise all information about the opponent model,
    making it the first robust-adaptation method to be able to do so in large imperfect-information games.
As an additional benefit, the use of matrix-valued states makes the algorithm simpler than traditional methods based on optimal value functions.
Our experimental results in poker and battleship show that \ABDshort{} yields more than a twofold increase in utility when facing opponents who make mistakes beyond the depth limit and also delivers significant improvements in utility and safety against randomly generated opponents.
\end{abstract}


\begin{document}

\input{macros}

%%% The following commands remove the headers in your paper. For final 
%%% papers, these will be inserted during the pagination process.

\pagestyle{fancy}
\fancyhead{}

%%% The next command prints the information defined in the preamble.

\maketitle 

%%%%%%%%%%%%%%%%%%%%%%%%%%%%%%%%%%%%%%%%%%%%%%%%%%%%%%%%%%%%%%%%%%%%%%%%

\input{sections/01-introduction}
\input{sections/02-background}
\input{sections/03-related_work}
\input{sections/04-looking_further}
\input{sections/05-experiments}

\input{sections/06-conclusion}

%%%%%%%%%%%%%%%%%%%%%%%%%%%%%%%%%%%%%%%%%%%%%%%%%%%%%%%%%%%%%%%%%%%%%%%%

%%% The acknowledgments section is defined using the "acks" environment
%%% (rather than an unnumbered section). The use of this environment 
%%% ensures the proper identification of the section in the article 
%%% metadata as well as the consistent spelling of the heading.

\begin{acks}
This research was supported by the Czech Science Foundation (grant no. GA22-26655S and grant no. GA25-18353S).
\end{acks}

%%%%%%%%%%%%%%%%%%%%%%%%%%%%%%%%%%%%%%%%%%%%%%%%%%%%%%%%%%%%%%%%%%%%%%%%

%%% The next two lines define, first, the bibliography style to be 
%%% applied, and, second, the bibliography file to be used.

\bibliographystyle{ACM-Reference-Format} 
\bibliography{bibliography}

%%%%%%%%%%%%%%%%%%%%%%%%%%%%%%%%%%%%%%%%%%%%%%%%%%%%%%%%%%%%%%%%%%%%%%%%

\end{document}

%% file: macros.tex
\newcommand{\MaVS}{matrix-valued state}
\newcommand{\MuVS}{multi-valued state}
\newcommand{\MaVSbig}{Matrix-Valued State}
\newcommand{\MuVSbig}{Multi-Valued State}
\newcommand{\MaVSshort}{MaVS}
\newcommand{\MuVSshort}{MuVS}

\newcommand{\ABDshort}{ABD}
    \renewcommand{\ABDshort}{\texttt{ABD}}
\newcommand{\ABD}{adapting beyond depth-limit}
\newcommand{\ABDbig}{Adapting Beyond Depth-Limit}
\newcommand{\ABDstart}{Adapting beyond depth-limit}

\newcommand{\SADshort}{ABD}
\newcommand{\SAD}{adaptating beyond depth-limit}
\newcommand{\SADbig}{Adaptating Beyond Depth-Limit}

\newcommand{\gain}{\textbf{gain}}
\newcommand{\exploitability}{\textbf{exploitability}}

\newcommand{\optline}{OPT}
\newcommand{\adaptingline}{ADPT}
\newcommand{\riskyline}{RISK}

% added by V

% Comments and other toggles:
\newboolean{commentsactivated}
\setboolean{commentsactivated}{false}
\setboolean{commentsactivated}{true}
\newcommand{\dm}[1]{\ifthenelse{\boolean{commentsactivated}}{{\color{blue} {\em DM: #1 }}}{}}
\newcommand{\vl}[1]{\ifthenelse{\boolean{commentsactivated}}{{\color{magenta} {\em VL: #1 }}}{}}
\newcommand{\vk}[1]{\ifthenelse{\boolean{commentsactivated}}{{\color{red} {\em VK: #1 }}}{}}

\newcommand{\Q}{\mathbb{Q}}
\newcommand{\N}{\mathbb{N}}
\newcommand{\R}{\mathbb{R}}
\newcommand{\E}{\mathbb{E}}

\newcommand{\mc}{\mathcal}
\newcommand{\mb}{\mathbb}
\newcommand{\defword}[1]{\textbf{\boldmath{#1}}}

\newcommand{\symbolPlaceholder}{{\, \cdot \, }}
\newcommand{\symbolPlaceholderBrackets}{{(\, \cdot \, )}}

\newcommand{\game}{{\mc{G}}}
\newcommand{\actions}{\mc A}
\newcommand{\action}{a}
\newcommand{\actionPl}{a_1}
\newcommand{\actionOpp}{a_2}
\newcommand{\pureStrategy}{s}
\newcommand{\policy}{\sigma}
\newcommand{\policyAlt}{{\policy'}}
\newcommand{\policies}{\Sigma}
\newcommand{\chance}{\textnormal{c}}
\newcommand{\playerSet}{\mc N}
\newcommand{\playerFunction}{pl}
\newcommand{\pl}{i}
\newcommand{\plAlt}{j}
\newcommand{\opp}{\others}          % opponent
\newcommand{\others}{{\textnormal{-}\pl}}
\newcommand{\Pl}{\textnormal{P\pl}}
\newcommand{\Plone}{\text{P1}}
\newcommand{\Pltwo}{\text{P2}}
\newcommand{\utility}{u}

\newcommand{\histories}{\mc H}
\newcommand{\history}{h}
\newcommand{\historyAlt}{g}
\newcommand{\prefix}{\sqsubset}                   % history prefix
\newcommand{\suffix}{\sqsupset}                   % history suffix
\newcommand{\gameroot}{\textnormal{root}}

\newcommand{\leaf}{z}
\newcommand{\leaves}{\mc Z}

\newcommand{\infoset}{I}
\newcommand{\infosetAlt}{\infoset'}
\newcommand{\infosetOpp}{J}
\newcommand{\infosetOppAlt}{\infosetOpp'}
\newcommand{\InfosetTree}{\mc I}
\newcommand{\infostate}{s}

\newcommand{\portfolio}{\mb P}
\newcommand{\pStrategy}{\pi}
\newcommand{\fixedOppStrategy}{\policy_2^\textnormal{fix}}

\newcommand{\CDBR}{\texttt{CDBR}}
\newcommand{\CDRNR}{\texttt{CDRNR}}
\newcommand{\CFR}{\texttt{CFR}}

% Vojta's new stuff
\newcommand{\RNRgame}{robust adaptation game}
\newcommand{\RAG}{\texttt{RbAdapt}}
\newcommand{\nOfSamples}{N_\textnormal{sample}}
\newcommand{\NE}{\textnormal{NE}}
\newcommand{\publicSet}{S}

%% file: sections/01-introduction.tex
\section{Introduction}

There has been a lot of progress on solving large imperfect information games,
    recently often enabled by depth-limited methods \cite{moravvcik2017deepstack,brown2018superhuman,brown2019superhuman}.
The high-level idea behind these methods is to avoid searching through the full game tree,
    which would be intractable.
Instead, depth-limited search only looks a few steps ahead
    and uses a value function -- typically a trained neural network -- to approximate the events that occur beyond the depth limit.
However, compared to perfect-information games such as chess and Go,
    the presence of imperfect information complicates depth-limited search in a number of ways.
    For example,
        being able to start the search from the current state of the game,
        as well as determining which value function to use,
        both become conceptually (and computationally) difficult \cite{burch2014solving}.
        
While most existing work is concerned with fully rational opponents,
    this assumption will often be false in practice,
    for example because the opponent is cognitively or computationally limited \cite{an2013deployed,nguyen2013analyzing}.
Several prior works investigate depth-limited methods for dealing with subrational opponents \cite{liu2022safe,milec2021continual,ge2024safe}.
However, these methods are conceptually quite challenging and the research on this topic seems far from complete.
In particular, one limitation of the existing methods is their inability to take advantage of the opponent's mistakes that only occur beyond the current depth limit.
    (As a contrived example, imagine that
        saying ``please let me win'' would always cause the opponent to give up in $d+1$ steps from now.
        If a method only looks $d$ steps ahead and assumes that the opponent plays optimally afterward, it will be unable to exploit this vulnerability.)

        % In game theory, rationality is often assumed when modeling players’ behaviors and strategies. However, in many practical applications, real-world opponents may deviate from full rationality due to bounded computational resources, incomplete information, or cognitive limitations. These subrational opponents pose unique challenges for decision-making systems that rely on traditional game-theoretic models, necessitating strategies that adapt to the opponent’s behavior \cite{hu2023can}.
        
        % A significant portion of game-theoretic research has focused on developing algorithms to deal with rational or near-rational players. However, in scenarios involving subrational agents, classical approaches may be inadequate. In complex games where the size of the game tree grows exponentially, the computational cost of finding an optimal solution becomes prohibitive. In such cases, depth-limited solving is often employed to reduce the computational burden by focusing on a limited portion of the game tree.

        % While depth-limited solving provides a practical approach to decision-making in large games, existing methods often assume a fixed strategy or use a single value function beyond the depth limit. This assumption can limit the ability to respond effectively to different subrational opponents. The challenge becomes how to model the opponent's strategy beyond the depth-limited portion while maintaining adaptive responses.

In this paper, we study the problem of \textbf{robust adaptation}
    -- that is, performing well against a known opponent that is not fully rational
        (represented by some fixed stochastic strategy)
    while not being overly vulnerable to unexpected opponents.
We are interested in solving this problem for large (two-player) \textbf{imperfect-information} games.
Since such games are typically too complex to compute the full strategy at once,
    we focus on the \textbf{online play} setting
    (where the game is being played in real time).
We operationalize the robust adaptation problem by making the conservative assumption that
    (rather than facing the fixed opponent all the time)
    there is some small probability $p$ that we are instead facing a \textit{worst-case} opponent (i.e., one that attempts to minimize our utility).
Importantly, this formulation of the problem does not require the knowledge of the opponent's preferences
    -- in other words, \textit{despite also being interested in general-sum games, it suffices to consider the problem for \textbf{zero-sum games}.}

Disregarding the issue of the prohibitive size of the game,
    robust adaptation has a straightforward solution \cite{johanson2008computing}:
    Consider an imperfect-information game that consists of two copies of the original game and an extra chance decision which secretly determines whether we face a fixed or fully rational opponent, then find a Nash equilibrium (NE) of this game
    (\Cref{fig:RNR_game}; for more details, see \Cref{sec:background}).
    Note that the fully rational opponent is aware that we are attempting to find a robust response to the fixed opponent, so their strategy will generally be different from the NE of the original game.
    This makes this problem more difficult than simply finding a best response to a convex combination of the fixed opponent and Nash equilibrium.
% Disregarding the issue of the prohibitive size of the game, this problem could be solved by computing a \textit{restricted Nash response} \cite{johanson2008computing}.
% This method (described in detail in \Cref{sec:background}) assumes that
%     player one is fully rational, but knows that with probability $p$, they will be facing an opponent who uses a fixed strategy $\fixedOppStrategy$.
%     However, with the remaining $1-p$ probability, they are facing a fully rational opponent, and the whole situation is common knowledge.
% This situation can be straightforwardly modelled as a game where player one has imperfect information about which opponent they are facing (\Cref{fig:background_DLS}).\

However, this approach would be intractable in large games,
    which means that the main difficulty lies in combining this method with depth-limited search.
% (One might intuitively imagine that the fully rational opponent will be playing a NE of the original game, which would greatly simplify the problem.
% However, this is not the case -- by assumption, the fully rational opponent will be best-responding to our strategy, which makes naive depth-limited approaches to robust adaptation inapplicable.)
To find a more scalable solution to the robust adaptation problem,
    we propose a novel approach we call \textit{Adapting Beyond Depth-limit} (\ABDshort{}).
Our point of departure is that the traditional approach to depth-limited solving is to call an optimal value function upon reaching the depth limit \cite{moravvcik2017deepstack,kovavrik2023value}.
In contrast, \ABDshort{} uses what we refer to as \textit{matrix-valued states}:
    upon reaching the depth limit, each player is given a choice of one strategy in their \textit{portfolio},
    after which both players receive the utility that corresponds to their joint choice \cite{brown2018depth,kovavrik2023value}.
The key insight behind \ABDshort{} is that by leveraging the \MaVS{} representation,
    we can emulate the sub-rational opponent by replacing their portfolio with the strategy given by their opponent model.

\textbf{Structure of the Paper:}
In \Cref{sec:related_work}, we review existing approaches to depth-limited solving and their limitations when applied to subrational opponents in large-scale games.
We then describe the intuitions which motivate the proposed approach (\Cref{sec:motivation})
    and follow with the algorithm's formal description and analysis (\Cref{sec:theory}).
Finally, we present experimental results showcasing the effectiveness of \ABDshort{} compared to traditional approaches (\Cref{sec:experiments}).

\textbf{Contributions:}
We provide the first method for depth-limited solving which is able to take advantage of all mistakes of the opponent -- including those beyond the depth limit  -- while remaining robust against rational opponents.
The proposed method is also much easier to train than comparable existing methods:
    instead of training a neural network to predict the values \textit{for every possible combination of reach probabilities at the depth limit} (of which there are infinitely many) \cite{kovavrik2023value,moravvcik2017deepstack},
    it suffices to predict the values corresponding to a limited portfolio of strategies \cite{kovavrik2023value,brown2018depth,brown2019superhuman}.
Our experimental results in poker and battleship show that \ABDshort{} yields more than a twofold increase in utility when facing opponents who make mistakes beyond the depth limit and also delivers significant improvements in utility and safety against randomly generated opponents.

% The previous approaches show that using a single value function is enough in many games to improve the results significantly. However, the authors show an artificially created game in which the adaptation worsens the solution compared to just following the Nash equilibrium. A real game with similar problems is Battleships. Both players first place a set of ships on a grid, then alternate moves and try to hit the other player's ships. The first player who sinks all the ships of the other player wins. Suppose we are playing against a player who will only shoot a specific cell after shooting everything else. If we start, we can always win by placing a ship to occupy this one cell. However, when we use depth-limited solving and we do not have the opponent strategy incorporated in beyond the depth-limit the adaptation lowers significantly as our depth-limit decreases with no possible exploitation when we only see the opponent next move, which would be placing his ships as well. We show detailed results in the Experiments section.

%% file: sections/02-background.tex
\section{Background}\label{sec:background}
To model imperfect information in games, we adopt the standard extensive-form model.

% \dm{add notion of exploitability and maybe gain (or we can just use utility instead)}
    % \vk{We might just stick with utility, unless there is some strong reason not to.}
\begin{definition}[EFG]\label{def:EFG}
A two-player zero-sum \textbf{extensive form game} (EFG) $\game = \left< \actions, \histories, \playerFunction, \pi_c, \mc I, \utility \right>$ is a tuple for which:
\begin{itemize}
    \item $\actions = \actions_1 \cup \actions_2$ are the sets of \textbf{actions} of each player
    \item $\histories$ is the set of legal \textbf{histories} in $\game$.
        Formally, $\histories$ is a finite tree on $\actions$
            (i.e., it consists of sequences of elements of $\actions$ and is closed under initial segments).
        Leaves of $\histories$ are called \textbf{terminal states} and denoted $\leaf \in \leaves$.
    \item $\playerFunction : \histories \setminus \leaves \to \{1, 2, \chance \}$ determines which player acts at $\history$.
    \item $\actions(h) := \{ a\in \actions \mid ha \in \histories \}$ denote the set of actions that are legal at the given state.
     \item For $\playerFunction(\history)=\chance$, $\policy_\chance(\history) \in \Delta (\actions(\history))$ is the \textbf{chance} strategy at $\history$.\footnote{
            Where $\Delta(X)$ denotes the set of all probability distribution over a finite set $X$.
        }
     \item $u : \leaves \to \R^2$ is a zero-sum \textbf{utility} function
        s.t. $\utility_2 = - \utility_1$.
     \item $\mc I = (\mc I_1, \mc I_2)$ represent the \textbf{information sets} (infosets) of each player.
        Formally, each $\mc I_\pl$ is partition of $\histories$
        that\footnotemark{} conveys enough information to identify player $\pl$'s legal actions.
            \footnotetext{
                In other words, we assume -- for convenience -- that the notion of player's information is defined everywhere in the game, rather than only when the player is acting.
            }
        We assume that each player has perfect recall.
        % \footnote{
        %     That is, for every $I\in \mc I_i$, $p(h)$ is either equal to $i$ for all $h\in I$ or for no $h\in I$. If for all, then $\actions(h)$ doesn't depend on the choice of $h\in I$.
        % }
\end{itemize}
\end{definition}

\noindent
We will abbreviate ``player one'' (resp. two) as \Plone{} and \Pltwo{}.

Each player plays the game using a \textbf{behavioural strategy} $\policy_\pl$.
Formally, $\policy_\pl$ is a mapping 
    % \begin{align*}
    $
        \policy_\pl
        :
        \left\{ \history \in \histories \mid \playerFunction(\history) = \pl \right\}
        \mapsto
        \policy_\pl(\history) \in \Delta(\actions(\history))
    $
    % \end{align*}
    which determines how to randomise over legal actions at every information set. $\Sigma_i$ is the set of all strategies of player $i$.
A strategy $\policy_\pl$ is said to be \textbf{pure} when
    each of the distributions $\policy_\pl(\history)$ is deterministic.
Since each pair of strategies $(\policy_1, \policy_2)$ induces a probability distribution over the set $\leaves$ of leaves the EFG,
    we can define the \defword{expected utility} in $\game$ as
    $\utility_\pl(\policy)
        := \E_{\leaf \sim \policy} \utility_\pl(\leaf)
    $.
Informally, the goal of each player is to maximise their expected utility.

This work will focus on identifying a strategy which performs well against a specific known opponent.
This is formalised using the notion of \textbf{best response} to an opponent strategy $\policy_2$
    -- i.e., a strategy $\policy_1$ for which
    $
        \utility_1(\policy_1, \policy_2)
        =
        \max_{\policy'_1} \,
            \utility_\pl(\policy'_1, \policy_2)
    $.
Another important concept is that of a \textbf{Nash equilibrium} (NE)
    -- formally, a pair of strategies $(\policy_1, \policy_2)$ is a NE if each of the strategies is a best response to the other strategy.

In a two-player zero-sum game, the \textbf{exploitability} of a strategy is the expected utility a fully rational opponent can achieve above the value of the game. Formally, exploitability $\mathcal{E}(\sigma_i)$ of strategy $\policy_i$ is
$
    \mathcal{E}(\policy_i) =  u_{-i}(\policy_i, \policy_{-i}) - u_{-i}(\policy^{NE}), \quad \policy_{-i} \in BR_{-i}(\policy_i).
$

\textbf{Safety} is defined based on exploitability and $\epsilon$-safe strategy is a strategy which has exploitability at most $\epsilon$.

We define \textbf{gain} of a strategy against a model as the expected utility we receive above the value of the game. We formally define the gain $\mathcal{G}(\policy_i, \policy_{-i})$ of the strategy $\policy_i$ against a strategy $\policy_{-i}$ as
$
    \mathcal{G}(\policy_i, \policy_{-i}) = u_i(\policy_i, \policy_{-i}) - u_{i}(\policy^{NE}).
$

\subsection{Robust Adaptation: Restricted Nash Response}\label{sec:sub:RNR}

Consider a situation where
    we believe that with probability $p \in [0, 1]$, our opponent in a game $\game$ will use some fixed strategy $\fixedOppStrategy \in \policies_2$,
    but with the remaining $1-p$ probability, they will play rationally (given that they know about our beliefs).
This situation can be formalised as looking for a \textbf{restricted Nash equilibrium} \cite{johanson2008computing},
    which is defined as a pair of strategies $(\policy^*_1, \policy^*_2)$ such that $\policy^*_2$ is a best response to $\policy^*_1$ and $\policy^*_1$ is a best response to $p \cdot \fixedOppStrategy + (1-p) \cdot \policy^*_2$.
The $\policy^*_1$ part of this pair is called $p$-\textbf{restricted Nash response} to $\fixedOppStrategy$.
It can be obtained by finding the (standard) Nash equilibrium
    of what we will refer to as the \textbf{\RNRgame{}} $\RAG(\game, p, \fixedOppStrategy)$:
    Structurally, $\RAG(\game, p, \fixedOppStrategy)$ consists of a chance node and two copies of $\game$ (\Cref{fig:RNR_game}).
    The chance node serves as the root of the game and its outcome is only observable to player two.
    With probability $p$, the chance node leads to a copy of $\game$ where
        \Pltwo{}'s decisions are replaced by chance nodes which select actions according to $\fixedOppStrategy$.
    With probability $1-p$, the chance node leads to a copy of the game that works identically to the original game $\game$
        (except that \Plone{} does not know which copy got selected).

\begin{figure}[tb]
    \centering
    \includegraphics[width=\linewidth]{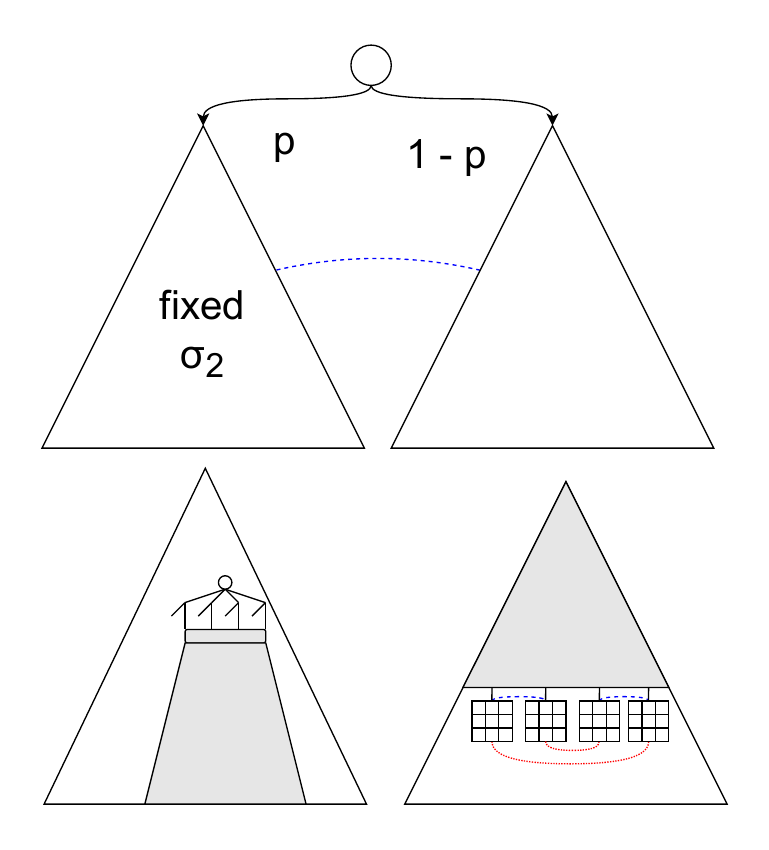}
    \caption{
        Top: Calculating restricted Nash response.
        Left: Starting search mid-play and a re-solving gadget.
        Right: Depth-limited solving using matrix-valued states.
    }
    \label{fig:background_DLS}\label{fig:RNR_game}
    \Description{Robust adaptation game, re-solving gadget, and depth-limited solving using matrix-valued states.}
\end{figure}

\subsection{Online Play in Large Games: Continual Resolving}\label{sec:sub:CR}

A common goal in studying large games is to find their Nash equilibrium, in the sense of being able to write down the whole strategy at once.
However, some games are too large for this to be tractable -- e.g., too large to even fit the full strategy into memory.
A more tractable alternative to this approach is \textbf{online play},
    which only aims to allow the user to \textit{play} the game according to some strategy.
    Formally, it aims to provide a mapping which
        takes an information set $\infoset \in \InfosetTree_1$ as an input
        and outputs a probability distribution $\policy^*_1 (\infoset) \in \Delta(\actions_1(\infoset))$
        for some $\policy^* \in \NE(\game)$ that is consistent across all $\infoset$.
    While such methods might require costly pre-training,
    the eventual computation of individual probability distributions $\policy^*_1(\infoset)$ must be efficient enough to work in a real time (e.g., not require more than a few seconds).

\textbf{Continual resolving} \cite{moravvcik2017deepstack,vsustr2019monte} is a method for online play in large imperfect-information games.
It consists of two key ingredients:
    The first is \textbf{re-solving}, which allows the algorithm to focus computation on the current position in the game,
        rather than having to start from the root each time.
    The second is \textbf{depth-limited solving}, which allows the algorithm to replace too-distant future events by an abstraction -- typically implemented as a value function.

\subsubsection{Restarting Search Mid-Game: Re-solving Gadget}\label{sec:sub:resolving}

The high-level idea behind re-solving is simply to replace the full game by a sub-game rooted at the current state.
While this is trivial in perfect-information games, extending this approach to imperfect information games requires some care.
First, to preserve the possibility of reasoning about uncertainty, the ``root'' of the game must include all histories from the current \textbf{public state} $\publicSet$ \cite{kovavrik2022rethinking}.
    (Where public states are defined as a partition of histories that is closed under the membership to both players' information sets.)
Second, the naive approach (called \textbf{unsafe re-solving})
    -- of replacing the initial part of the game by a single chance node,
    with the probability of each $\history \in \publicSet$ being proportional to its reach probability under some Nash equilibrium --
    unfortunately generally results in an exploitable strategy \cite{burch2014solving}.
    (Intuitively, this is because \Plone{} might make mistakes that can only be punished by taking non-equilibrium actions in the part of the game that already took place.)
This issue can be avoided by augmenting the root of the game with a \textbf{re-solving gadget} \cite{burch2014solving,moravcik2016refining}
    which re-introduces \Pltwo{}'s ability to punish \Plone{}'s suboptimal play by unobserved actions prior to the current state.
% We will use the approach proposed by \citet{burch2014solving}, which requires the knowledge of the strategy of \Plone{} and values of \Pltwo{} in histories $\history \in \publicSet$.
We will use the approach proposed by \citet{milec2021continual}, which only requires the knowledge of the strategy of \Plone{}.

\subsubsection{Depth-Limited Solving via Value Functions}\label{sec:sub:value_functions}

A common approach to depth-limited solving is to only consider a limited \textbf{look-ahead tree},
    in which every history $\history$ at the \textbf{depth limit} $d \in \N$ is turned into a terminal state.
The utilities of these new leaves are then given by some \textbf{optimal value function} $v$ which assumes that, beyond the depth limit, both players play according to some Nash equilibrium $\policy^*$.
In practice, the value function is typically implemented by a neural network trained to approximate the exact values \cite{moravvcik2017deepstack}.
As a complicating factor in imperfect-information games, the values of optimal values functions depend not only on the current history $\history$, but also on the strategy that the players used prior to reaching $\history$ \cite{kovavrik2023value}.
This makes such value functions more costly to both train and use during deployment.

% An important complication is that most games that are of practical interest are too large to be computationally tractable.
% For example, it might even be impossible to explicitly store a player's full strategy in memory.
% To overcome these difficulties, researchers sometimes use \textbf{depth-limited methods},
%     which reduce the computational complexity in games
%     by restricting the analysis to a limited portion of the game tree (up to some depth $d$).
%     Beyond this depth, the game is approximated by a value function $v$,
%         typically implemented by a neural network trained to approximate the exact values.
%
% One approach to depth-limited solving is to assume that beyond the depth limit, both players are going to play (approximately) optimally \cite{kovavrik2023value}.
%     However, the presence of imperfect information means that what constitutes an optimal play depends on the probability distribution over the (unobserved) events that already occurred.
%     This in turn means that when relying on an optimal value function,
%         even the computation of a best response becomes a non-trivial matter,
%         requiring multiple passes of the depth-limited game tree.

\subsection{Depth-Limited Solving via Matrix-Valued States}\label{sec:sub:MVS}

An alternative approach to depth limited solving
    (first described by \citet{brown2018depth}, used by \citet{brown2019superhuman}, and further studied by \citet{kovavrik2023value})
    relies on a construction that resembles the normal-form representation of EFGs.
The idea behind the classical normal-form representation of an EFG is that instead of acting sequentially, the players specify their full strategy ahead of time \cite{MAS}.
However, one can also consider a depth-limited analogy of this concept,
    where the agents act sequentially for the first $d$ steps of the game,
    and then only specify their full strategy for the \textit{remainder} of the game.
In the idealised version of this approach,
    each player would be able to chose from \textit{all} of their pure strategies.
In practice, this would be intractable,
    so each player is only allowed to chose from a more limited \textbf{portfolio} $\portfolio_\pl \subset \policies_\pl$.
% Since this method can be implemented by replacing states at the depth-limited by simultaneous decisions -- i.e., matrices --
%     we refer to it as \textbf{matrix-valued states}.
% The term \textbf{multi-valued states} \cite{brown2018depth} refers to the important special case
    % where the portfolio of one of the players consists of only a single strategy.
Depending on the quality (i.e., representativeness) of the portfolios, this method can recover an exact -- or at least approximate -- NE of the game \cite{brown2018depth, kovavrik2023value}.

\begin{definition}[Matrix-valued depth-limited game]
Let
        $\game$ be a two-player EFG,
        $d \in \N$
        and $\portfolio_1 \subset \policies_1$, $\portfolio_2 \subset \policies_2$.
    The corresponding \textbf{matrix-valued depth-limited game}
        $\game_d(\portfolio_1, \portfolio_2)$
    is defined as the EFG that works as follows:
    For histories of length $< d$, $\game_d(\portfolio_1, \portfolio_2)$ works identically to $\game$.
    For histories of length $d$, the information sets of each player are the same as information sets in $\game$.
    However, unlike in $\game$,
        both players are active in these histories
        and their legal actions correspond to selecting $\pStrategy_1 \in \portfolio_1$, resp. $\pStrategy_2 \in \portfolio_2$.
    This yields payoffs
        $
            \utility^{\game_d(\portfolio_1, \portfolio_2)}_\pl (\history, \pStrategy_1, \pStrategy_2)
            :=
            \E [ \utility_\pl(\leaf) \mid \leaf \sim(\pStrategy_1, \pStrategy_2), \leaf \textnormal{ extends } \history ]
        $
        and terminates the game.
\end{definition}

\noindent
Using this terminology,
    the existing methods \cite{brown2018depth, brown2019superhuman,kovavrik2023value} can be understood as finding an \NE{} of particular games $\game_d(\portfolio_1, \portfolio_2)$.
The name ``\textbf{matrix-valued states}''
    is a reference to the multi-valued states used by \citet{brown2018depth}
    and comes from the fact that in $\game_d(\portfolio_1, \portfolio_2)$,
    the histories at depth $d$ can be viewed as matrices of size $|\portfolio_1| \times | \portfolio_2|$. 

% However, this does not allow us to take advantage of the knowledge of the opponent's strategy $\fixedOppStrategy$.
% This is partially remedied in the \CDRNR{} algorithm \cite{milec2021continual},
%     which assumes that the opponent uses $\fixedOppStrategy$ before the depth limit in the fixed part of the RNR tree.
%     However, \CDRNR{} still uses the optimal value function, which is analogous to assuming that $\portfolio_\pl = \policies_\pl$
%     -- i.e., that beyond the depth limit, the players can choose from among all legal strategies.

%% file: sections/03-related_work.tex
\section{Related Work}\label{sec:related_work}

Although opponent modeling in games has been extensively studied \cite{nashed2022survey}, this paper takes a different approach. We assume that an opponent model is already available and focus instead on developing strategies for safe adaptation against this model.
We discuss three lines of related work:
    adapting in games that are sufficiently small to disregard computational issues,
    adaptation in larger but perfect-information games,
    and adaptation in large games with imperfect information.
Notably, all of works cited in this section either solve the full game
    -- which is impossible for the larger games --
    or rely on depth-limited methods which do not allow for incorporating information about the opponent model.
In \Cref{sec:sub:CDRNR}, we describe CDRNR, the state of the art algorithm used as a baseline in our empirical evaluation.

\subsection{Adapting to Opponents in Smaller Games}

A key feature of smaller games is that we can always compute an exact best response.
Consequently, there are many works which focus on on creating the opponent model and then adapting to it by shifting to the best response \cite{carmel1996learning,billings2004game,southey2012bayes,ganzfried2018bayesian}.

% A notion that puts emphasis on safety is the restricted Nash response.
There is also a number of other approaches.
\citet{johanson2008computing} show that by using restricted Nash response against a fixed strategy $\sigma^F$, they can reconstruct a Pareto optimal set with respect to \gain{} and \exploitability{}.
\citet{milec2021continual} also show an explicit bound on exploitability as a function of the parameter $p$.
Another approach that is is similar to restricted Nash response is to focus on using limited data to exploit data-biased response \cite{johanson2009data}. Data-biased responses use a parameter for each the opponent's decisions and can differ depending on our confidence in the data.
Finally, \citet{ganzfried2015safe} focus on safe adaptation in repeated game,
    where they show that one can safely improve their utility by tracking the profit gained from opponent mistakes and risking at most that amount.

\subsection{Adapting to Opponents in Perfect Information Games}

Early work on adaptation in perfect-information games started by creating simple opponent models returning values for actions and using search to adapt using those values \cite{carmel1994m}.
Later works focused on analyzing moves that are still part of the optimal strategy but adapt to the opponent by winning faster against moves that might be losing against a perfect player, hinting at the notion of safety \cite{sen1997learning}.
This line of work concluded with the authors of \cite{markovitch2005learning} creating a method that not only computes where one can adapt to the opponent but also keeps a self-weakness model and tries to search for actions where the opponent is weak, but the method will be strong.

\subsection{Adaptation in Large Imperfect Information Games}

In large games with imperfect information, \citet{timbers2020approximate} propose an optimised reinforcement learning that allows us to compute an exact response.
% In large games with imperfect information, we can compute an exact response using reinforcement learning, and authors in \cite{timbers2020approximate} propose an optimized reinforcement learning method exactly for that.
However, this approach requires a significant investment time for each opponent, making it inapplicable when we need to respond to a new opponent in real time.
\citet{wu2021l2e} propose a method that can adapt online by evolving a neural network that can quickly adapt during play.
However, this method does not provide any safety guarantees.
Another recent contribution is the algorithm Cicero, developed for the large multi-player game Diplomacy, which involves adapting to the assumed intent of the opponents as its crucial component \cite{meta2022human}.
However, because most actions (i.e., every except communication) in Diplomacy are public, the direct applicability of this algorithm to general imperfect-information games is limited.

\subsubsection{Safe Online Adaptation}

An important next step in adaptation to opponents is the ability to adapt in online play while having safety guarantees.
Several works tackle this issue by combining restricted Nash response with continual re-solving \cite{moravvcik2017deepstack}, with main differences being in the choice of the re-solving gadget (cf. \Cref{sec:sub:resolving}) \cite{liu2022safe,milec2021continual,ge2024safe}. 
\citet{liu2022safe} use the max-margin gadget and only use the opponent model as it reaches the subgame.
\citet{ge2024safe} use the re-solving gadget of \citet{burch2014solving}.
Since neither of these works uses the opponent strategy in the re-solving subgame,
    they have strong safety guarantees but have a lower ability to adapt to the opponent.

\subsubsection{Continual Depth-limited Responses}\label{sec:sub:CDRNR}
The current state-of-the-art approach for depth-limited opponent-robust exploitation is the \textbf{continual depth-limited restricted Nash response (CDRNR)} \cite{milec2021continual}
    and its special case \textbf{continual depth-limited best response (CDBR)}
    (which corresponds to setting CDRNR's ``robustness'' to zero).
On the high level,
    CDRNR combines a restricted Nash response with depth-limited solving
    by using a neural network to approximate the optimal play after the depth limit
    (i.e., disregarding the opponent model there).
The authors highlight limitations of existing re-solving gadgets, demonstrating why restricted Nash response guarantees fail with them.
To address this, they introduce a \textbf{full gadget} that preserves the path from the root to the current subgame, fixing the resolving player and allowing opponent deviations.

% In multiplayer games, the opponent adaptation is crucial for strong performance.
% This was well illustrated by the recent results in the game Diplomacy where adapting to the assumed intent of the opponents is an essential part of the algorithm \cite{meta2022human}.
% In diplomacy, all the moves are public, so there is no need to adapt to what the opponent would do in the later turn. However, adaptation will be crucial in games where more information is hidden.

%% file: sections/04-looking_further.tex
\section{Motivation}\label{sec:motivation}
When playing against a subrational opponent, we can follow the optimal strategy and limit our \gain{} by a lot. We can adapt to some opponent mistakes that do not force us to deviate from the optimal strategy, and finally, we can adapt to opponent mistakes in cases where we need to deviate from the optimal strategy.

The concepts smoothly overlap in imperfect-information games, and it is a scale of how much we want to adapt and how much we would leave ourselves to be exploited. In perfect information games without chance, it is discrete, and we show an example of a tic-tac-toe to explain the concept.

\begin{example}[Tic-Tac-Toe]
Figure~\ref{fig:tick-tac-toe} shows different situations in tic-tac-toe. Assume that the optimal strategy player 1 is playing always starts in the middle. In the left column, we see what happens when he starts in the middle. The opponent reacts by playing in the corner, and then we will have a game where he reacts optimally, ending in a draw. However, if player 1 starts in the corner, the opponent would follow as is shown in the middle column of Figure~\ref{fig:tick-tac-toe} and would play a suboptimal move, resulting in his loss (player 1 will win by playing any corner, forcing block and following with the other corner setting up two possible threes and the opponent can only block one). This deviation for player 1 is still part of another optimal strategy, and if played optimally by both players, the game would end in a draw. Finally, in the right column of Figure~\ref{fig:tick-tac-toe}, we have a situation where the opponent starts and plays in the middle. We will have a draw if player 1 plays the optimal move in any corner. However, if player 1 plays a move that would be losing against the optimal player, the opponent will malfunction and lose. Knowing that player 1 can win when surely facing the opponent but can be severely punished when facing anyone else.

\ABDshort{} has a parameter $p$ which governs how much we want to adapt to the opponent (corresponding to the probability $p$ in RNR). In this case, we have three options: the optimal line (\optline{}), the safe adapting line (\adaptingline{}), and the risky adapting line (\riskyline{}). When $p = 0$, player 1 will just reconstruct his optimal strategy, which in our case would be \optline{}. As soon as $p > 0$, player 1 would switch to \adaptingline{} but would still not play the \riskyline{}. At $p > 0.5$, player 1 would switch from \optline{} to \riskyline{} but would still prefer \adaptingline{} over \riskyline{} if the two were in the same decision. Finally, when we set $p = 1$, player 1 will no longer distinguish between \adaptingline{} and \riskyline{}.
\end{example}

We consider the setting,
    previously explored by \cite{milec2021continual,liu2022safe,ge2024safe},
    where the goal is to
        play (approximately) optimally against a fixed opponent $\fixedOppStrategy$,
        % compute the (approximate) best strategy against a fixed opponent,
        given a model that predicts their play. At the same time, we assume the model can be inaccurate, and we need to limit our exploitability.

In smaller games, this problem can be solved exactly by finding a restricted Nash response, which can be parametrized to recover the Pareto set of epsilon-safe best responses.

The leading method, presented by \cite{milec2021continual},
    is based on depth-limited solving
    and assumes that both players behave rationally beyond the depth limit.
    % uses depth-limited solving,
    %     building the game tree only a few moves ahead.
% At the depth limit, a value function assumes both players behave rationally for the rest of the game.
This approach, however, has a critical flaw:
    it cannot take advantage of the opponent's mistakes that happen after the depth limit.
For illustration, consider the following example:

\begin{example}[Simplified Battleships]\label{expl:simplified_battleships}
Take a simplified 2x2 version of Battleships with a single 1x1 ship as shown in Figure~\ref{fig:battleships_example}.
The opponent acts second and fires uniformly, except for the top-left corner, which they always shoot last.
To exploit this, it suffices to place our ship in the top-left, allowing us to consistently win the game.
However,
    suppose we can only look one move ahead in our tree
    and that we assume that the opponent will be shooting optimally.
    Then, we will only have the actions used to place ships in our depth-limited tree for both players. Therefore, we will be unable to exploit the opponent.
\end{example}

\begin{figure}
    \centering
    \includegraphics[width=0.6\linewidth]{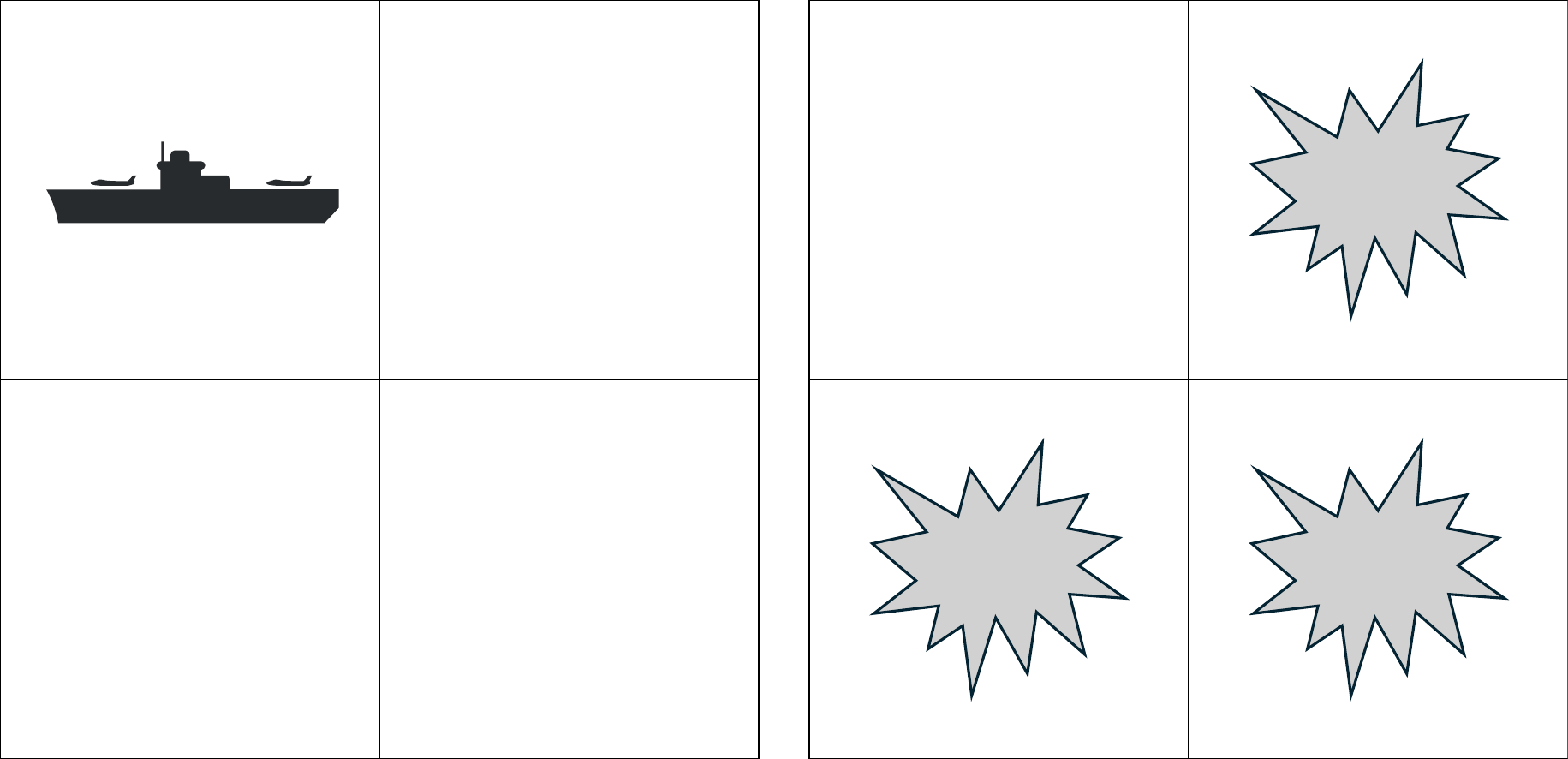}
    \caption{Setup for \Cref{expl:simplified_battleships}. Battleships of size 2x2 with one 1x1 ship. Fixed opponent will not fire to the top left corner, and our best response is to place the ship there.}
    \label{fig:battleships_example}
\end{figure}

% \section{Looking Beyond the Depth-limit}
\section{Opponent Exploitation Beyond the Depth-limit}\label{sec:theory}

In this section, we describe an online game-playing method, \emph{\ABD{}} (\ABDshort{}), that is able to strike a balance between playing well against a specific opponent and remaining unexploitable.
Unlike previous methods that rely on depth-limited solving, \ABDshort{} will be able to take advantage of all of the opponent's mistakes, irrespective of whether they occur before or beyond the depth limit.
In \Cref{sec:sub:algorithm_ideal_case}, we describe the idealised version of the method that assumes access to a value function that captures the behaviour of the specific opponent.
In \Cref{sec:experiments}, we describe two practical methods for approximating the idealised value function described in \Cref{sec:sub:algorithm_ideal_case}.
We then conclude in \Cref{sec:sub:br} by showing how the algorithm can be simplified when we are certain about the opponent's identity and do not need to worry about remaining unexploitable.

\subsection{Opponent Exploitation in Online Play via Matrix-Valued States}\label{sec:sub:algorithm_ideal_case}

If we disregard efficiency, the \textit{idealised} version of the \ABD{} algorithm can be described as follows:
\begin{itemize}
    \item Let $\game$ be a game and $\emptyset \neq \portfolio_\pl \subset \policies_\pl$ portfolios for each player.
    \item The parameters of the algorithm are
        a depth limit $d \in \N$,
        a fixed opponent strategy $\fixedOppStrategy \in \policies_2$,
        and $p \in [0, 1]$ which controls the adaptation-exploitability tradeoff.
    \item The input to the \ABDshort{} algorithm is an infoset $\infoset \in \InfosetTree_1$.
    \item The output, obtained in a way described below, is a probability distribution $\policy_1(\infoset) \in \Delta(\actions_1(\infoset))$.
    \item The algorithm considers the game $\game^3$, which can be obtained as follows:
    \begin{enumerate}
        \item Let $\game^1$ be the \RNRgame{} corresponding to $\game$, $\fixedOppStrategy$, and $p$.
        \item Let $\game^2$ be the re-solving subgame \cite{milec2021continual} corresponding to \Plone{} starting from $\infoset$ in $\game^1$.\footnotemark{}
        \item Let $\game^3 := \game^2_d(\portfolio_1, \portfolio_2)$ be the matrix-valued depth-limited variant of $\game^2$,
            corresponding to looking $d$ steps ahead from $\infoset$, after which each player is limited to only using strategies from their portfolio $\portfolio_\pl$.
            (Note that, as in $\game^1$ and $\game^2$, \Pltwo{} is still limited to using $\fixedOppStrategy$ in the half of the game tree where their strategy is meant to be fixed.)
    \end{enumerate}
    \item To compute $\policy_1(\infoset)$, the idealised algorithm finds a Nash equilibrium $\policy^*$ of $\game^3$ and sets $\policy_1(\infoset) := \policy^*_1(\infoset)$.
\end{itemize}
    \footnotetext{
        Strictly speaking, building the resolving gadget requires additional information about the game
        -- namely, the strategy of \Plone{} in the preceding part of the game \cite{milec2021continual}.
        When using the \ABDshort{} algorithm for online play, this strategy is always available from the previous call of the algorithm.
    }

This algorithm is idealised in three ways:
First, it assumes we have access to a method that can compute an exact Nash equilibrium of $\game^3$.
    This can, in theory, be achieved by methods such as linear programming \cite{koller1992complexity}.
    However, this would often be impractically slow in large games.
    As a result, a practical implementation of \ABDshort{} will use an approximate iterative algorithm such as counterfactual regret minimisation \cite{zinkevich2008regret}.

Second, when constructing the depth-limited version of $\game^2$,
    the idealised algorithm does not take advantage of the specific structure of the game.
    This is not an issue in the part of the game where \Pltwo{} can use any of the strategies from their portfolio
        (since the corresponding values
            $\E [ \utility_\pl(\leaf) \mid \leaf \sim(\pStrategy_1, \pStrategy_2), \leaf \textnormal{ extends } \history ]$,
            $\pStrategy_\pl \in \portfolio_\pl$,
        can be pre-computed without the knowledge of $\fixedOppStrategy$).
    However, in the part of the game where \Pltwo{} must use their fixed strategy $\fixedOppStrategy$, the construction requires the knowledge of
        $\E [ \utility_\pl(\leaf) \mid \leaf \sim(\pStrategy_1, \fixedOppStrategy), \leaf \textnormal{ extends } \history ]$.
    In this section, we will (unrealistically) stick with the assumption that this value can be calculated exactly every time it is needed.
    A more realistic approach to this obstacle will be discussed in \Cref{sec:sub:algorithm_practical}.

Finally, the idealised algorithm also does not shed any light on the shape of the re-solving game $\game^2$.
    This could pose a problem if it turned out that the shape of $\game^2$ (and thus also $\game^3$) is very sensitive to the choice of $\fixedOppStrategy$ and $p$, requiring a separate construction each time.
    Fortunately, this is not the case:
    We will now show that by relying on the specific structure of the \RNRgame{} $\game^1$,
    we can derive a general formula which can be used to define $\game^2$.

The gadget at the top of $\game^3$ closely follows the construction described by \cite{milec2021continual}. Figure~\ref{fig:abd_full_example} 
illustrates the gadget’s structure. In the subgame where the opponent is fixed, we introduce an initial chance node, which sets the initial reaches at the root of the subgame to the combined reaches of all players, including chance.

In the subgame where the opponent plays rationally, we use the full gadget proposed by \cite{milec2021continual}. The full gadget preserves the path from the root of the game to the subgame. A value function replaces the trajectories leading outside of the subgame. The authors in \cite{milec2021continual} used poker-specific value functions. Since these value functions could only be queried after each round, the authors needed to keep the whole previous poker round as the gadget. We avoid this limitation by constructing \MaVS{} for the first public state of each alternative trajectory, significantly reducing the gadget's size and removing the need to evaluate the neural network representing the value function in each round.

\begin{restatable}{proposition}{propIdealCase}\label{prop:ideal_case}
    Let
        $\game$ be a two-player zero-sum EFG,
        $\fixedOppStrategy \in \policies_2$ be a fixed strategy of the opponent,
        and $p \in [0, 1]$.
    When
        $
            \portfolio_\pl
            \supseteq
            \left\{
                \pureStrategy_\pl \in \policies_\pl
                \mid
                \pureStrategy_\pl \textnormal{ is pure undominated}
            \right\}
        $,
    the strategy produced by the \ABDshort{} algorithm is
        % a Nash equilibrium of the \RNRgame{} corresponding to $\game$, $\fixedOppStrategy$, and $p$.
        a $p$-restricted Nash response to $\fixedOppStrategy$ in $\game$.
\end{restatable}

\begin{proof}
    From Theorem 4 in \cite{kovavrik2023value}, it follows that when the portfolios $\portfolio_\pl$ contain at least all pure undominated strategies, the algorithm coincides with Continual Resolving (applied to the \RNRgame{}) using full gadget of \citet{milec2021continual}.
    By Theorem 5.3 in \cite{milec2021continual}, we will converge to restricted Nash response when we use the full gadget.
\end{proof}

% \propIdealCase*       % uncomment this to repeat the statement (to be placed in appendix)

% EXPLAIN(\SADshort{} algorithm assumes the opponent uses  $\fixedOppStrategy$ both before and after the depth limit. We construct a RNR game and set a depth limit the same way \CDRNR{} does. Then when we reach the depth limit in the fixed game $G^F$ we set the continuations as values of fixed opponent strategy against the portfolio of the player. In the free game $G'$ we create a continuation where the first player chooses from it's portfolio and then the second one chooses, resulting in a value given by continuing in the game using the respective portfolios.

% When we move in the game deeper we reconstruct subgame from a public state and use the same approach as \CDRNR{} of creating a gadget above the $G'$ and using unsafe resolving for $G^F$. This method used with all pure strategy continuation as portfolio allows us to recover the RNR.

% \begin{proposition}[With full portfolio, \ABDshort{} finds an exact best response]
%     \label{prop:full_pl_one_portfolio}
% Let $\game$ be a two-player zero-sum game
%     and let $\fixedOppStrategy \in \policies_2$ be a fixed strategy of the opponent.
% When $\portfolio_1 = \{ \pureStrategy_1 \in \policies_1 \mid \pureStrategy_1 \textnormal{ is pure} \}$,
%     the best-response to the pre-depth limit portion of $\fixedOppStrategy$ in $\game_d(\portfolio_1, \{ \fixedOppStrategy\} )$
%     is a best-response to $\fixedOppStrategy$ in $\game$.
% \end{proposition}

\begin{figure*}
    \centering
    \includegraphics[width=\linewidth]{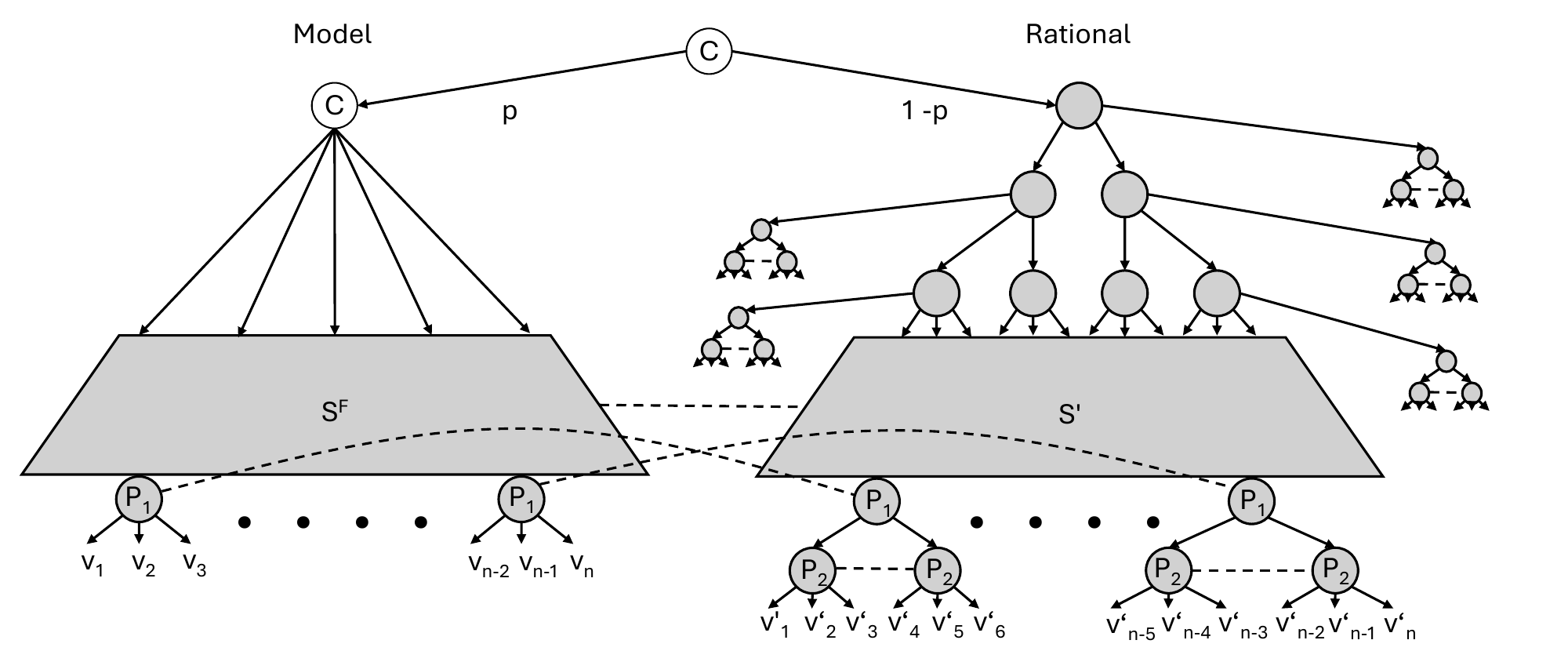}
    \caption{
        An illustration of the game used by the \ABD{} algorithm.
        The upper part is the gadget, whose presence is necessary when not running the algorithm from the root.
        Left: fixed opponent.
        Right: rational opponent.
    }
    \label{fig:abd_full_example}
    \Description{The game used by the \ABD{} algorithm.}
\end{figure*}

\subsection{Computing Opponent-Specific Values}\label{sec:sub:algorithm_practical}

% In larger games, the idealised form of the \SADshort{} algorithm would be intractable for two reasons.
% First, the number of all pure-strategy continuations would be intractable.
%     (For example, in a mid-sized Battleships game with 5x5 board and two 2x2 ships and depth limit $d=1$, there are $10^{25}$ continuations.)
% The solution to this problem is to use a limited portfolio of different strategies.
%     (For example in battleships we can use shooting even spaces and shooting odd spaces and then some ship sinking continuation.)
%     Authors in \cite{kubicek2023look} propose a method to learn such portfolios during training.

In larger games, the idealised form of the \SADshort{} algorithm would be intractable for two reasons.
First, if we wanted to rely on \Cref{prop:ideal_case}, we would need to use a portfolio which consists of all pure undominated strategies in $\game$.
    In practice, many of the strategies are redundant
    -- when using a portfolio strategy to evaluate a particular state, the only relevant part of the strategy is its restriction to the plausible future states (not its past actions or actions in states that are no longer possible).
    While this greatly reduces the size of the number of strategies,
        even the reduced portfolios would be intractably large.
    (For example, in a mid-sized Battleships game with a 5x5 board, two 2x1 ships, and depth limit $d=1$, the portfolios would contain roughly $10^{25}$ strategies.)
The solution to this problem is to use a limited portfolio of different strategies
    (some of which might not necessarily be pure).
\citet{kubicek2023look} propose a method to learn such portfolios during training.

% Second, even with a manageably-sized portfolio $\portfolio_1$,
%     obtaining the exact values corresponding to $(\pStrategy_1, \fixedOppStrategy)$, $\pStrategy_1 \in \portfolio_1$, could be prohibitively expensive.
% We solve this problem by sampling from $(\pStrategy_1, \fixedOppStrategy)$ to approximate the correct values.

Second, even with a manageably-sized portfolio $\portfolio_1$,
    the future part of the game tree might be large,
    making it prohibitively expensive to obtain the exact values of $(\pStrategy_1, \fixedOppStrategy)$ and $(\pStrategy_1, \pStrategy_2)$, $\pStrategy_\pl \in \portfolio_\pl$.
For $(\pStrategy_1, \pStrategy_2)$,
    these values might be obtained prior to learning $\fixedOppStrategy$
    or even come as a side-product of constructing $\portfolio_\pl$.
However, $(\pStrategy_1, \fixedOppStrategy)$ cannot be pre-computed in such matter.
As a baseline variant, we thus \emph{compute the values
    $
        \E [ \utility_\pl(\leaf) \mid \leaf \sim(\pStrategy_1, \fixedOppStrategy), \, \leaf \textnormal{ extends } \history ]
    $
    by sampling trajectories $\leaf$ from using $\fixedOppStrategy$ (and $\pStrategy_1 \in \portfolio_1$)}.
    More specifically, we will additionally parametrize the algorithm by some $\nOfSamples$ and replace each of the values above with the corresponding empirical mean.

Our approach significantly reduces computational costs compared to the previously used neural network-based value functions. Instead of repeatedly querying a neural network during each iteration of the solving process, we sample the values only at the start of the solving phase and then apply the solving algorithm on the well-defined extensive form game. In contrast, neural network methods must invoke the value function at each iteration for depth-limited leaves, which can become computationally expensive, especially with large network architectures.

\subsection{Special Case: Efficient Best-Response Computation}\label{sec:sub:br}

When we are \emph{certain} that the opponent will play the fixed strategy $\fixedOppStrategy$ or we want to evaluate the exploitability of $\fixedOppStrategy$ we will have the $p=1$.\footnotemark{}
    \footnotetext{
        I.E., compared to CDBR from \cite{milec2021continual}, the value function is no longer changing.
    }
This allows us to simplify the computation in two ways:

% EXPLAIN(When we set the algorithm to fully exploit we can disregard the part where the opponent can act, which in turn creates a game where the opponent is fixed everywhere (compared to CDBR from \cite{milec2021continual} the value function is no longer changing) and we can simplify the computation.)

\paragraph{Observation 1: finding a best response in a single pass.}
When the opponent is fixed both before and beyond the depth limit, we obtain an EFG where one player is fixed everywhere.
This means that we no longer need to rely on iterative approaches to game solving
    (e.g., the \CFR{} algorithm \cite{zinkevich2008regret} used in \CDBR{} (cf. \Cref{sec:sub:CDRNR}) necessary to accommodate the value function represented by a neural network).
Instead, we are able to compute a best response using backward induction,
    which only requires a single pass of the look-ahead tree
    and thereby saves significant computation time.

\paragraph{Observation 2: searching a smaller part of the look-ahead tree.}
When we perform depth-limited search later in the game,
    the presence of imperfect information requires us to consider more states than just the (unknown) true state of the game.
In particular, approaches such as \CDBR{} must start the search from all histories in the current \textit{public state} \cite{kovavrik2022rethinking}.
In contrast, the use of a fixed opponent even after the depth-limit allows \ABDshort{} to only start the search from the histories in the current \textit{information set} of player one.
(E.G., in poker, this would mean the difference between considering all possible private hands for both players vs only considering all possible hands for the opponent.)

\begin{figure}
    \begin{minipage}{0.5\linewidth}
        \includegraphics[width=0.31\linewidth]{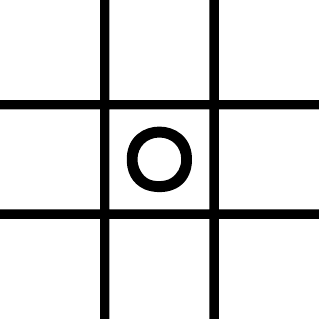}
        \hfill
        \includegraphics[width=0.31\linewidth]{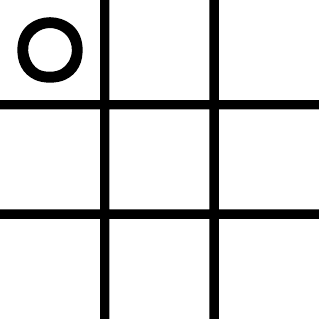}
        \hfill
        \includegraphics[width=0.31\linewidth]{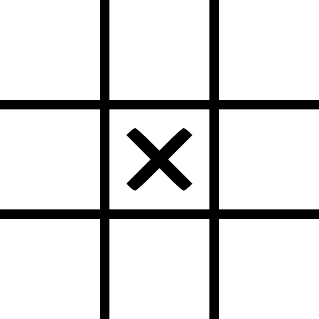}
    \end{minipage}

    \vspace{5pt}
    
    \begin{minipage}{0.5\linewidth}
        \includegraphics[width=0.31\linewidth]{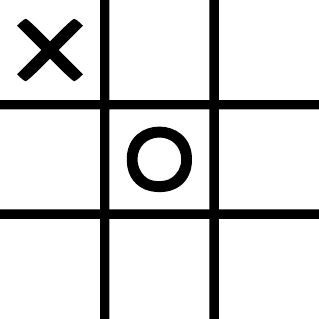}
        \hfill
        \includegraphics[width=0.31\linewidth]{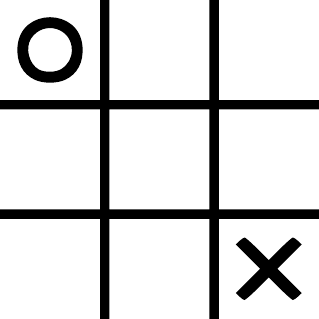}
        \hfill
        \includegraphics[width=0.31\linewidth]{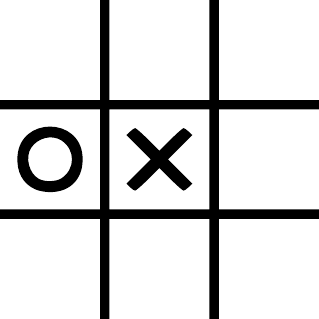}
    \end{minipage}    
    \caption{Example situations in Tic-tac-toe. The adapting player is the circle. Left: Opening with (one of) optimal strategies and reaction with the only optimal move. Middle: Opening with another optimal strategy and reaction, which is losing. Right: Opening by the opponent optimally and the suboptimal reaction of adapting player, knowing the opponent will then lose.}
    \label{fig:tick-tac-toe}
    \Description{Tic-tac-toe.}
\end{figure}

%% file: sections/05-experiments.tex
\section{Experiments}\label{sec:experiments}
We conduct experiments in both small and large imperfect information games to highlight the failure modes of CDBR (\Cref{sec:sub:CDRNR}) and demonstrate the adaptability of our method against subrational strategies.

\subsection{Failure of Previous Methods}
The first experiment aims to illustrate the failure modes of the previous method, specifically CDRNR (CDBR). We utilize a 2x2 Battleships game featuring one 1x1 ship, as outlined in Example 1. The opponent adopts a uniform shooting strategy, except for the top-left corner, which they target only as a last move. In Table~\ref{tab:winrate_depth}, we present the expected value of CDBR at increasing depths. With a game value of 0.25, it is evident that CDBR fails to exploit the opponent’s shooting pattern when only the first move of ship placement is observed. When CDBR observes the first shooting move, it is still not enough to shift its strategy, and only after observing two out of three strategically significant moves can it find the best response.

In contrast, our method consistently achieves victory in this setup by employing a portfolio of four strategies, each avoiding a different cell while shooting uniformly otherwise. We compute the values for the portfolios exactly in this experiment.

\begin{table}[htbp]
    \centering    
    \begin{tabular}{|c|c|c|c|c|c|}
        \hline
        \textbf{Depth}   & 1    & 2     & 3    & 4    \\ \hline
        \textbf{CDBR} & 0.25 & 0.25 & 1.00 & 1.00 \\ \hline
        \textbf{\ABDshort{} ($p = 1$)} & 1.00 & 1.00 & 1.00 & 1.00 \\ \hline
    \end{tabular}
    \caption{Winrate at different depths on 2x2 battleships with one 1x1 ship against an opponent who always shoots the upper left corner last. Depth is the number of future opponents' moves contained in the depth-limited subgame.}
    \label{tab:winrate_depth}
\end{table}

\subsection{Varying $p$ in Battleships}
The next experiment explores the performance of \ABDshort{} with a fixed depth of 2, varying the parameter p to examine its impact on the trade-off between exploitability and gain. We use the same 2x2 Battleships game described in Example 1 with the same portfolio, and we compute the value exactly again. CDRNR uses an optimal value function computed by a linear program. Unlike the previous method, which struggles to exploit this opponent, we show the tradeoff between \exploitability{} and \gain{} in Figure~\ref{fig:battleships_abd_results}. We report the robust results this way since it is not that important which $p$ corresponds to which points but how close we are to the Pareto optimal set considering \gain{} and \exploitability, which the RNR shows. Unsurprisingly, we see that CDRNR cannot achieve any \gain{} and \ABD{} is mimicking RNR exactly in this simple game.

\begin{figure}
    \centering
    \includegraphics[width=\linewidth]{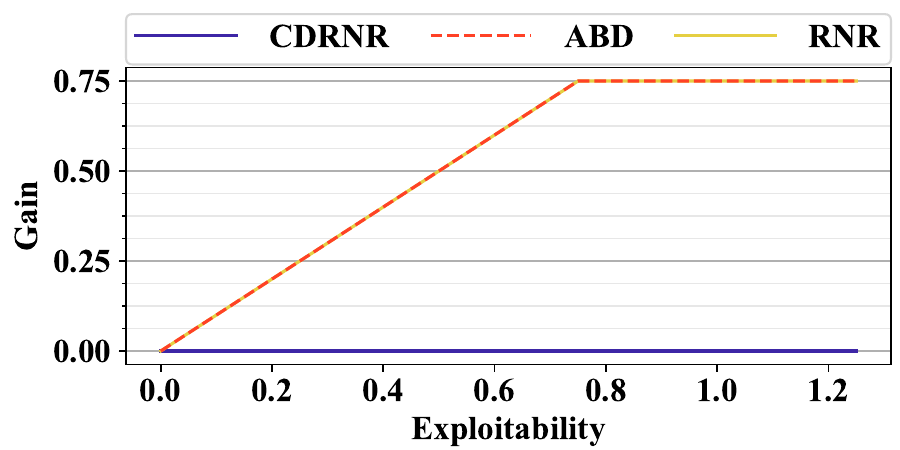}

    \caption{Results of \ABDshort{} in Battleships against an opponent who does not shoot to the top left corner. ABD coincides with RNR, and CDBR is not able to exploit the strategy with depth = 2}
    \Description{Results in Battleships.}
    \label{fig:battleships_abd_results}
\end{figure}

\subsection{Varying $p$ in Leduc Hold'em}
In this experiment, we evaluate the performance of \ABDshort{} in a more realistic setting using the Leduc Hold’em poker game. We test two types of opponents: a uniform opponent, who plays without any particular strategic bias, and an exploitable opponent, who exhibits predictable tendencies. Specifically, the exploitable opponent only checks or folds in the first round and only bets or calls in the second round. We will denote the opponent's strategy S1.

Player 1 uses a portfolio of strategies consisting only of a tight passive strategy (which only involves checking and folding) and a loose aggressive strategy (which consists of betting and calling when a bet is impossible). Values at the depth limit are sampled using 10 samples per history.

Our experiment aims to assess how \ABDshort{} adapts to these distinct opponent behaviors across various settings of p. Figure~\ref{fig:leduc_abd_results} shows the \exploitability{} and the \gain{} achieved by player 1 to compare performance against both the uniform and the exploitable opponent. We see that \ABD{} cannot reconstruct an unexploitable strategy because of the limited portfolio. However, against both players, the \ABD{} is able to recover strategies overall closer to the optimal RNR set.

We see that the main difference is in the ability to exploit the S1. This is expected since we are now using all the information about the opponent and not only the information available before the current depth-limit. Since this is the most interesting part of the comparison, our next experiments will be focused on those extreme cases where $p = 1$.

\begin{figure}
    \centering
    \includegraphics[width=\linewidth]{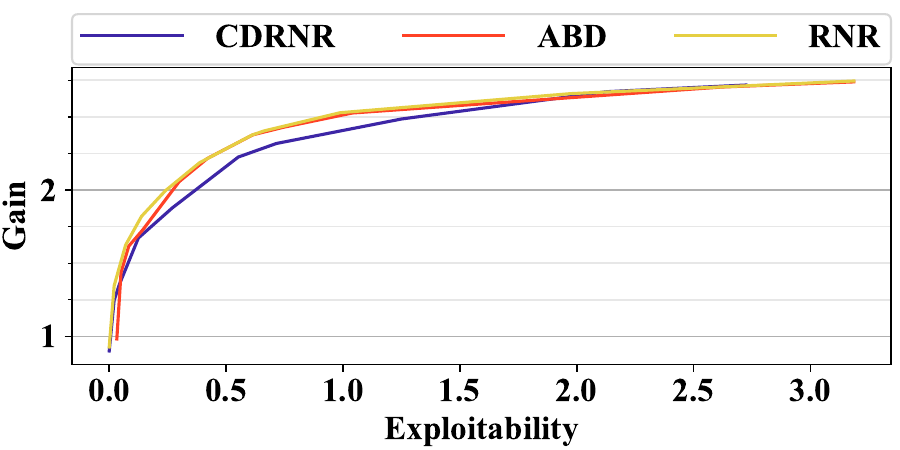}
    
    \includegraphics[width=\linewidth]{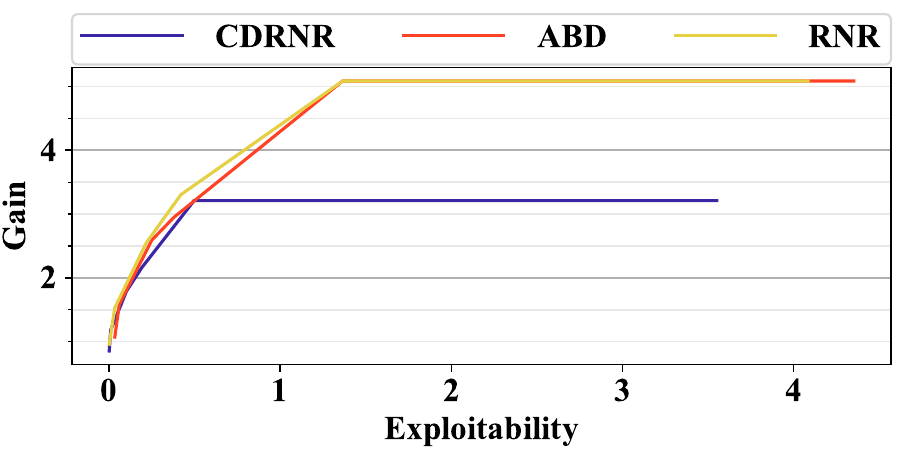}
    \caption{Results of \ABDshort{} in Leduc hold'em against (top) uniform opponent and (bottom) opponent playing S1. A comparison is made between \ABDshort{}, the previous method (CDRNR), and the theoretical upper limit (RNR).}
    \label{fig:leduc_abd_results}
    \Description{Results in Leduc hold'em.}
\end{figure}

\subsection{More Leduc Strategies}
Next, the Leduc experiment uses the same setup and strategies; $S1$ is the same as in the previous experiment. Also, we add strategies: $S_2$, which operates in a reverse manner, always bets or calls in the first round and always folds or checks in the second round; $S_3$, which is based on $S_1$ but folds in the second round when only folding or calling is available; and $S_4$, which mirrors $S_2$ with the same folding modification.

We evaluate \ABD{} against the four opponent strategies alongside CDBR and show the results in Table~\ref{tab:leduc_results_cdbr_abd}. Additionally, we conducted experiments against 1000 randomly generated strategies denoted Random, where \ABDshort{} significantly outperforms CDBR, although by a much smaller margin. We report a 95\% confidence interval.

\begin{table}[htbp]
    \begin{tabular}{|l|c|c|c|c|c|}
        \hline
        & S1 & S2 & S3 & S4 & Random \\ \hline
        \textbf{CDBR} & 1 & 5 & 1 & 3 & 2.475 ± 0.016 \\ \hline
        \textbf{\ABDshort{} ($p = 1$)} & 2.3 & 5 & 4.2 & 5 & 2.536 ± 0.015 \\ \hline
    \end{tabular}
    \caption{Results of CDBR and \ABDshort{} against S1 to S4 and Random in Leduc hold'em.}
    \label{tab:leduc_results_cdbr_abd}
\end{table}

\subsection{Large Game}
Our final experiment examines 5x5 Battleships with two 2x2 ships, which have over $10^{50}$ possible histories. The opponent places ships uniformly and shoots uniformly, with the exception of the top-left corner, which they target only as a last move. We chose a uniform shooting strategy to avoid giving our method an unfair advantage that might arise from a more deterministic opponent. This configuration allows us to exploit the opponent’s behavior by placing ships in the top left corner, which they will never target.

Our approach utilizes a portfolio of three strategies: one that shoots uniformly across the grid and two that target even and odd cells, respectively. We use 100 samples per history to evaluate leaves in the depth-limited tree and compute the strategy with depth limit 2. \ABDstart{} correctly identifies placing the ship in the top left corner as the best response.

To assess the necessary samples in more noisy settings, we trained against an opponent who can occasionally target the top-left corner (with a 5\% probability); we analyzed the fraction of games in which our method successfully identified the best response over 100 trials and reported the results in Table~\ref{tab:winrate_samples}. This indicates that utilizing 100 samples in this configuration, which was our initial setting, is more than sufficient for achieving reliable performance.

\begin{table}[tb]
    \centering
    \begin{tabular}{|l|c|c|c|c|c|}
        \hline
        \textbf{Samples} & 1 & 2 & 3 & 4 & 5 \\ \hline
        \textbf{Correct} & 83\% & 88\% & 95\% & 98\% & 100\% \\ \hline
    \end{tabular}
    \caption{Portion of correctly recovered best responses with different sample sizes.}
    \label{tab:winrate_samples}
\end{table}

%% file: sections/06-conclusion.tex
\section{Conclusion}

This paper addressed the challenge of safely adapting to subrational opponents in large, imperfect information games where computational constraints necessitate depth-limited solving. Previous approaches rely on simplified assumptions about the opponent’s strategy beyond the depth limit, which can hinder the ability to respond effectively. To overcome these limitations, we introduced a novel framework that leverages matrix-valued states to represent the opponent's strategy beyond the depth limit.

We showed that in terms of robust adaptation to random opponents, our method slightly outperforms the previous state of the art method, CDRNR. However, when facing opponents who make mistakes in later parts of the game, our method significantly outperforms CDRNR, with a twofold increase in achieved \gain{}.

Theoretical results show that the idealized setting can recover the theoretical restricted Nash response -- the theoretical optimum -- for the given opponent.
We also identified two important simplifications when $p$ -- the confidence in the opponent model -- is set to one, allowing for a significant speed up in computation.